\documentclass[a4paper, UKenglish, pdfa]{lipics-v2021}

\usepackage{amsmath, amssymb}
\usepackage{epigraph}
\usepackage{graphicx}
\usepackage{tikz}
\usetikzlibrary{arrows, automata, positioning, arrows.meta, calc, decorations.pathreplacing}

\newcommand{\N}{\mathbb{N}}
\newcommand{\Z}{\mathbb{Z}}
\newcommand\Alt{\textrm{Alt}}
\newcommand\Sofic{\textrm{Sofic}}
\newcommand\supp{\textrm{supp}}
\newcommand{\vstack}[2]{{\genfrac{}{}{0pt}{}{#1}{#2}}}
\DeclareMathSymbol{\rightarrow}{\mathord}{symbols}{"21}

\title{Subshifts Defined by Nondeterministic and Alternating Plane-walking Automata}
\titlerunning{Nondeterministic and Alternating Plane-walking Automata}

\date{}
\author{Benjamin {Hellouin de Menibus}}{Université Paris-Saclay, CNRS, Laboratoire Interdisciplinaire des Sciences du Numérique, 91400, Orsay, France \and \url{https://www.lisn.upsaclay.fr/~hellouin/}}{hellouin@lisn.fr}{https://orcid.org/0000-0001-5194-929X}{}
\author{Pacôme Perrotin}{Université Paris-Saclay, CNRS, Laboratoire Interdisciplinaire des Sciences du Numérique, 91400, Orsay, France \and \url{https://www.pacomeperrotin.com}}{pacome.perrotin@gmail.com}{https://orcid.org/0000-0003-1197-2676}{}

\authorrunning{B. Hellouin de Menibus and P. Perrotin}

\funding{This work received financial support from the ANR-22-CE40-0011 project Inside Zero Entropy Systems.}
\Copyright{Benjamin Hellouin de Menibus and Pacôme Perrotin}



\nolinenumbers
\acknowledgements{Many people have contributed to this article through discussions, suggestions and failed attempts; we wish to thank (in alphabetical order) Florent Becker, Amélia Durbec, Pierre Guillon and Guillaume Theyssier. We are grateful to Ville Salo for providing the proof of Theorem~\ref{thm:soficAlt} and pointing us towards the Kari-Moore rectangles of Section~\ref{sec:squares}, and to Denis Kuperberg and Thomas Colcombet for pointing us towards works on tree-walking automata cited in Section~\ref{sec:treewalking}.} 

\EventEditors{Olaf Beyersdorff, Micha\l{} Pilipczuk, Elaine Pimentel, and Nguyen Kim Thang}
\EventNoEds{4}
\EventLongTitle{42nd International Symposium on Theoretical Aspects of Computer Science (STACS 2025)}
\EventShortTitle{STACS 2025}
\EventAcronym{STACS}
\EventYear{2025}
\EventDate{March 4--7, 2025}
\EventLocation{Jena, Germany}
\EventLogo{}
\SeriesVolume{327}
\ArticleNo{56}

\begin{document}

\maketitle

\setlength\epigraphrule{0pt}
\epigraph{\textit{Il voulut être sofic,\\ il ne fut que pompé.}}{}

\begin{abstract}
     \emph{Plane-walking automata} were introduced by Salo \& Törma to recognise languages of two-dimensional infinite words (subshifts), the counterpart of $4$-way finite automata for two-dimensional finite words. We extend the model to allow for nondeterminism and alternation of quantifiers. We prove that the recognised subshifts form a strict subclass of sofic subshifts, and that the classes corresponding to existential and universal nondeterminism are incomparable and both larger that the deterministic class. We define a hierarchy of subshifts recognised by plane-walking automata with alternating quantifiers, which we conjecture to be strict. 
\end{abstract}

\section{Introduction}

One-dimensional finite automata are a classical model to recognise languages of finite words. They have been extended to recognise languages of finite patterns in two and more dimensions, often called \emph{picture languages}, starting with the work of Blum and Hewitt in 1967~\cite{blum1967}. While the one-dimensional model is very robust to changes in definition, this is not the case in higher dimensions and many different models have been introduced with varying computational power; see \cite{kari2011survey} for a survey that focuses on the non-deterministic case.

Symbolic dynamics are concerned with subshifts, which are languages of infinite words or patterns. In dimension 1, sofic subshifts can be seen as the infinite counterparts to regular languages, and have three equivalent definitions: the set of infinite walks on a labelled graph (finite automaton without initial nor final states); the set of infinite words avoiding a regular set of forbidden subwords; the letter-to-letter projection of a subshift of finite type. The latter definition carries through to higher dimensions without difficulties to define higher-dimensional sofic subshifts.

There are many ways to extend regular languages to higher dimensions.
Some models such as \emph{$4$-way automata} keep the notion of a linear \emph{run} where the head reads the input pattern, letter by letter, by moving ("walking") over the pattern; this contrasts with models such as \emph{recognisable languages} where acceptance consists in checking local constraints instead of a run. Broadly speaking, the latter models are more powerful than the former, and the same phenomenon arises for languages on trees: tree automata vs. tree-walking automata. \cite{giamma97} provides a catalogue of the different kinds of models, while a more recent survey on the "walking" models can be found in \cite{kari2011survey}.
  
When applying these models to subshifts, recognisable languages yields an alternative definition of sofic subshifts; this was first done in \cite{jonoska09} to our knowledge, and we present this construction in Section \ref{sec:soficautomata}. Recently, Salo and Törma introduced \emph{plane-walking automata} (PWA)~\cite{planewalking}, a particular case of graph-walking automata, which are a counterpart of $4$-way automata for infinite patterns. In particular, they proved that deterministic plane-walking automata define a class of subshifts that is strictly between subshifts of finite type and sofic subshifts, extending to infinite patterns a result from \cite{kari2001new} on $4$-way automata. It is also proved in \cite{kari2001new} that, for finite patterns, nondeterministic $4$-way automata are strictly more powerful than the deterministic version and that languages of alternating $4$-way automata are incomparable with recognisable languages, a sharp contrast with the one-dimensional case. 

In this paper, we introduce nondeterminism and alternations to plane-walking automata and consider the classes of higher-dimensional subshifts obtained this way. We prove that subshifts accepted by $\Sigma_1$ PWA (existential nondeterminism) and $\Pi_1$ PWA (universal nondeterminism) form incomparable classes (Theorem~\ref{thm:main}), both strictly larger than the deterministic case, and that subshifts accepted by arbitrary alternating PWA still form a strict subclass of sofic subshifts (Theorem~\ref{thm:soficAlt}). This yields two new classes of higher-dimensional subshifts that are intermediate between finite type and sofic subshifts, and natural in that they generalise one-dimensional sofic subshifts. We introduce an alternating hierarchy of nondeterministic power from deterministic up to unbounded alternating plane-walking automata, and conjecture that this hierarchy of subshifts does not collapse; to our knowledge, this is not known even for finite words. Our proofs involve equivalents of the pumping lemma for two-dimensional infinite patterns.

Definitions and results are written with the two-dimensional case ($\Z^2$) in mind, although they extend easily to any higher dimension, and our definitions also extend to any finitely generated groups (see \cite{groupwalking}).

\section{Configurations and Subshifts}

\subsection{Symbolic Dynamics}

We call \emph{positions} the elements of $\Z^2$, which is endowed with the Manhattan distance $d$.
We use $\rightarrow\ = (1,0)$ and $\uparrow\ = (0,1)$, $\bullet = (0,0)$ and so on; we often write $0$ as a position instead of $(0,0)$.

Let $\Sigma$ be a finite set called \emph{alphabet}. A \emph{configuration} $x$ is an element of $\Sigma^{\Z^2}$, while a \emph{pattern} $p$ is an element of $\Sigma^S$ for some finite $S\subset \Z^2$, called the support of $p$ and denoted $\supp(p) = S$. For $i \in \Z^2$, $\sigma^i(x)$ is a new configuration shifted by $i$, i.e. $(\sigma^i(x))_{j} = x_{i+j}$. In dimension one, we call a pattern a \emph{word}.

Given $\pi:\Sigma' \to \Sigma$, we extend $\pi$ to $\Sigma'^{\Z^2}$ by putting $\pi(x)_p = \pi(x_p)$ for all $p\in\Z^2$. 

A \emph{subshift} $X$ is a set of configurations that is closed in the product topology and invariant by all shifts; more concretely, it is defined as the set of configurations where no pattern from a set of forbidden patterns $\mathcal F$ appears.

A subshift defined by a finite set of forbidden patterns is called \emph{of finite type} (SFT for short). By a standard technique of higher block coding, replacing the alphabet $\Sigma$ by $\Sigma^{[0,n]\times [0,n]}$ for $n$ large enough, we can assume without loss of generality that forbidden patterns all have support $\{\bullet,\rightarrow\}$ or $\{\bullet,\uparrow\}$, and we do throughout the paper.

A subshift that can be written as $\pi(X)$, where $X$ is an SFT on alphabet $\Sigma'$ and $\pi:\Sigma' \to \Sigma$, is called \emph{sofic}. In dimension one, sofic subshifts have alternative definitions as the set of bi-infinite walks on some labelled graph or as the set of bi-infinite words avoiding some regular set of forbidden words.

We denote the classes of SFT and sofic subshifts \textrm{SFT} and \textrm{Sofic}, respectively.
  
\subsection{Two-dimensional Automata}  

We define an abstract model of two-dimensional automata. We use different notions of acceptance in the paper, and add additional restrictions to the model as necessary. 

\begin{definition}
  A two-dimensional automaton on $\Z^2$ is a labelled directed graph $A = (V,E,\Sigma, D, I)$, where
  \begin{itemize}
  \item $V$ and $E$ are finite sets of vertices and edges, respectively, where $E \subset V^2\times \mathbb Z^2$ (we call the second component the direction);
  \item $\Sigma$ is a finite alphabet and $D : V \to \Sigma$ associates a symbol to each vertex. 
  \item $I \in V$ are the initial states and $D$ is a bijection from $I$ to $\Sigma$. For any $a\in\Sigma$, we denote $i_a$ the only state in $I$ such that $D(i_a)=a$.
  \end{itemize}

If the automaton is \emph{alternating}, then there is additionally a map $Q : V \to \{\forall, \exists\}$ associating a quantifier to each vertex.
\end{definition}

Notice that since $E$ is finite, the set of possible directions is a finite subset of $\Z^2$.
  
\subsection{Recognisable Picture Languages}\label{sec:soficautomata}

Recognisable languages, as introduced in \cite{reco}, are a possible extension of regular languages to higher-dimensional words (or \emph{picture languages}) as projections of local languages, which can be expressed using two-dimensional automata.

The same model applied on subshifts yields an alternative definition of sofic subshifts using two-dimensional automata. 

\begin{definition}\label{def:recognising}
Let $A$ be a non-alternating two-dimensional automaton whose set of directions is $\{\uparrow, \rightarrow\}$. Given a pattern or a configuration $x$, an \emph{recognising run} of $A$ on $x$ is a function $r : \supp(x) \to V$ such that:

\begin{itemize}
\item for all $i \in \supp(x), D(r(i)) = x_i$;
\item for all $i \in \supp(x)$ such that $i + \rightarrow \in \supp(x)$, there is an edge $(r(i), r(i + \rightarrow), \rightarrow)\in E$, and similarly for $\uparrow$. 
\end{itemize}
Then:
\begin{itemize}
\item the recognised language $R_f(A)$ is the set of all patterns that admit a recognising run.
\item the recognised subshift $R_\infty(A)$ is the set of all configurations that admit a recognising run.
\end{itemize}
\end{definition}

Notice that this definition is intrinsically  nondeterministic in the choice of the run and makes no use of initial states, so we omit $I$ in this section. The first equivalence of the following result appeared in \cite[Proposition 7]{jonoska09}; we provide a short proof in our framework.

\begin{proposition}\label{prop:soficautomata}
For a subshift $X$, the following are equivalent:
\begin{enumerate}
\item $X$ is sofic;
\item $X = R_\infty(A)$ for some automaton $A$;
\item $X$ is defined by a set of forbidden patterns $R_f(A)^c$ for some automaton $A$,
\end{enumerate}
where $A$ is assumed to satisfy the hypotheses of Definition~\ref{def:recognising}.
\end{proposition}

\begin{proof}
$(1\Rightarrow2)$ Let $Y$ be a SFT on alphabet $\Sigma'$ given by a finite set $\mathcal F$ of forbidden patterns and $\pi : \Sigma'\to\Sigma$ be such that $\pi(Y) = X$; remember that we assume that patterns in $\mathcal F$ have support $\{\bullet, \uparrow\}$ or $\{\bullet, \rightarrow\}$. We define $A = (V,E,\Sigma',D)$ by setting $V = \Sigma'$ and $D = \pi$, and $E$ is defined as follows: for all $p\in \Sigma'^{\{\bullet, \uparrow\}}$, $(p_\bullet, p_\uparrow, \uparrow)\in E$ if and only if $p\notin \mathcal F$, and similarly for $\rightarrow$. By construction, $y \in \Sigma'^{\Z^2}$ is a recognising run of $A$ on $x$ if and only if $y\in Y$ and $\pi(y) = x$, so $X = R_\infty(A)$.

$(2\Rightarrow1)$. Let $A = (V,E,\Sigma',D)$ be an automaton and define the SFT $Y \subset V^{\Z^2}$ by the set of forbidden patterns:
\[
\mathcal F = \{p\in V^{\{\bullet,\rightarrow\}}:(p_\bullet, p_\rightarrow, \rightarrow)\notin E\}) \cup \{p\in V^{\{\bullet,\uparrow\}} : (p_\bullet, p_\uparrow, \uparrow) \notin E\}.\]
Again by construction, $y$ is a recognising run of $A$ on $x$ if and only if $y\in Y$ and $D(y) = x$, so $X = D(Y)$ is sofic.

$(2\Leftrightarrow 3)$ The two statements are equivalent for any automaton $A$. If $x \in R_\infty(A)$ has a recognising run, then the restriction of this run to any finite pattern in $x$ is recognising, so all patterns in $x$ belong to $R_f(A)$. Conversely, if all patterns in $x$ admit a recognising run, then $x$ does as well by a standard compactness argument.
\end{proof}

Notice that the third condition involves complementation, in contrast with the one-dimensional case; recognisable languages are not closed by complement in dimension 2 \cite{anselmo10}.

\section{Plane-walking Automata and Associated Subshifts}

Plane-walking automata generalise the definition of one-dimensional sofic subshifts seen as the set of infinite walks on a labelled graph.

\subsection{Definitions}

The automata we use in this section  correspond to the plane-walking automata (PWA) from Salo and Törma \cite{planewalking} with one head and added nondeterminism. They are alternating two-dimensional automata with a specific acceptance condition, that we call alternating plane-walking automata to be consistent with the literature.

\begin{definition}[Subshifts defined by plane-walking automata]\label{def:PWA}
Let $A = (V,E,\Sigma,D,I,Q)$ be an alternating PWA. Given $x \in \Sigma^{\Z^2}$, $p\in\Z^2$ and $v\in V$, there is an accepting run on $x$ starting from $(p, v)$ if $D(v) = x_p$ and:

\begin{itemize}
\item $Q(v) = \exists$ and there is an edge $(v, v', d) \in E$ and an accepting run starting from $(p+d, v')$, or
\item $Q(v) = \forall$ and all edges $(v, v', d) \in E$ with $D(v') = x_{p+d}$ have an accepting run starting from $(p+d, v')$; furthermore, there must be one such edge.
\end{itemize}

A configuration $x$ is \emph{accepted} by $A$ if, for every $p\in\Z^2$, there is an accepting run of $A$ on $x$ starting from $(p, i_{x_p})$. The set of configurations accepted by $A$ is a subshift, denoted $L_\infty(A)$. 
\end{definition}

The lack of a base case in the previous definition means that an accepting run must be infinite; in other words, a run accepts if it never reaches a position and state with no possible move. The model of Salo and Törma used rejecting states, which we opted not to do by symmetry w.r.t the lack of accepting states; this is a stylistic choice as a rejecting state can be simulated by a state with no outgoing edge, and an accepting state by a state with a loop with direction $\bullet$.

An alternating plane-walking automaton and its associated subshift are illustrated in Figure~\ref{even_automata_figure}. Without loss of generality, by adding additional states, the set of directions can be restricted to $\{\uparrow, \downarrow, \leftarrow, \rightarrow, \bullet\}$ which we assume in the rest of this article.

\begin{figure}[t]
  \centering
  \begin{minipage}{0.49\textwidth}
    \centering
    \begin{tikzpicture}  [->, node distance = 1cm]
  \node[state] (a) {$1$};
  \node[state] (b) [right =of a] {$1$};
  \node[state] (f) [above =of a] {$1$};
  \node[state] (g) [right =of f] {$0, \forall$};
  \node[state] (c) [right =of g] {$1$};
  \node[state] (d) [above =of c] {$0$};
  \node[state] (e) [left =of d] {$1$};

  \node (input0) [above =.5 of a] {};
  \node (input1) [above =.5 of g] {};
 	\path
      (input0) edge (a)
      (input1) edge (g)
      (g) edge node [left] {$\rightarrow$} (b)
      (c) edge node [right] {$\rightarrow$} (d)
      (b) edge [bend right = 20] node [right] {$\rightarrow$} (c)
      (c) edge [bend right = 20] node [left] {$\rightarrow$} (b)
      (g) edge node [below] {$\uparrow$} (f)
      (e) edge node [below] {$\uparrow$} (d)
      (f) edge [bend right = 20] node [right] {$\uparrow$} (e)
      (e) edge [bend right = 20] node [left] {$\uparrow$} (f)
      (g) edge node [below right, xshift = -0.2cm, yshift = 0cm] {$\uparrow \rightarrow$} (d);

	\draw (a) to [out=20, in=-20, looseness=7] node [right] {$\bullet$} (a);
	\draw (d) to [out=20, in=-20, looseness=7] node [right] {$\bullet$} (d);
\end{tikzpicture}
  \end{minipage}
  \begin{minipage}{0.39\textwidth}
    \centering
    \begin{tikzpicture}  [->]
  \foreach \x in {1, ..., 10} {
    \foreach \y in {1, ..., 8} {
      \node[gray] at (0.4*\x, 0.4*\y) {$0$};
    }
  }

  \foreach \x\y in {
    1/7, 3/7, 1/6, 3/6,
    7/7, 6/6, 8/6, 6/5, 8/5,
    7/4, 5/3, 5/2
  } {
    \fill[white!90!black] (0.4*\x - 0.2, 0.4*\y - 0.2) rectangle ++(0.8, 0.4);
    \node[black] at (0.4*\x, 0.4*\y) {$1$};
    \node[black] at (0.4*\x + 0.4, 0.4*\y) {$1$};
  }

  \draw (0, 0) rectangle (0.4 * 11, 0.4 * 9);
\end{tikzpicture}
  \end{minipage}
  \caption{\label{even_automata_figure}
    On the left, an example of an alternating plane-walking automaton.
    On the right, a finite pattern of a configuration accepted by
    the automaton. The associated subshift is the subshift where all vertical and horizontal runs of $1$s are
    either of even size or infinite.
  }
\end{figure}

We denote $\Alt$ the class of subshifts $X$ such that $X = L_\infty(A)$ for some alternating PWA $A$.
It is not difficult, as in Proposition~\ref{prop:soficautomata}, to extend the notion of acceptance to finite patterns (considering a run as accepting when it leaves the support of the pattern), where it coincides with alternating 4-way picture-walking automata. Subshifts in $\Alt$ can equivalently be defined by a set of forbidden patterns that are the complement of the language of some alternating PWA. As for recognisable languages, these languages are not closed by complement \cite[Theorem 7]{kari2011survey}.

Plane-walking automata as considered by Salo \& Törma \cite{planewalking} are deterministic, in the sense that there is only one outgoing edge from each state on which the run does not fail immediately; therefore the quantifiers in the definition are unused. We denote by $\Delta_1$ the class of subshifts defined by such deterministic plane-walking automata, and we call deterministic every state where the quantifier is not relevant, so we omit it in the pictures for clarity.

\begin{definition}[Branch, Footprint]
A run on $x$ can be represented by a tree where each vertex corresponds to a current position and state.

A \emph{branch} of a run is a branch in that tree, in the usual sense. 

The \emph{footprint} of a subtree or a branch is the set of all visited positions.
\end{definition}

\subsection{Comparison with SFT and Sofic Subshifts}

We show that subshifts defined by alternating plane-walking automata are an intermediate class between the two classical classes of SFT and sofic subshifts. Notice that, in dimension one, the classical powerset construction tells us that added nondeterminism does not impact the power of the finite automata, so every class defined in this article coincide with sofic subshifts.

The next result is proved in \cite{planewalking} (Inclusion is stated without proof,  Strictness is Lemma~1; Inclusion also follows from Lemma~\ref{lem:inter-SFT} below).

\begin{proposition}
$SFT \subsetneq \Delta_1\subset \Alt$.
\end{proposition}

The following result appeared in \cite{kari2001new} (Theorem~1) for finite patterns. We translate the proof in our framework since our statements differ due to differing models.

\begin{proposition}
  $\Alt\subset \Sofic$.
\end{proposition}
\begin{proof}
Let $A = (V,E,\Sigma, D, I, Q)$ be an alternating PWA. Let $\Sigma' = \Sigma \times \mathcal P(V)$, where a symbol from $\Sigma'$ encodes the set of states starting from which there is an accepting run in the current position. Let $\pi_1$ and $\pi_2$ be the projections on the first and second component, respectively.
    
We define a SFT $Y \subset \Sigma'^{\Z^2}$ by the following set $\mathcal F$ of forbidden patterns: for $p \in \Sigma'^{\{\leftarrow, \rightarrow, \uparrow, \downarrow, \bullet\}}$, denote $p_\bullet = (s, S)$ where $s\in \Sigma$ and $S\in \mathcal P(V)$. Then $p\in \mathcal F$ if and only if:

\begin{enumerate}
 \item $\exists v\in S, D(v) \neq s$, or 
 \item $S\cap I = \emptyset$, or
 \item there is $v\in S$ such that
    \begin{itemize}
    \item $Q(v) = \exists$ and for all $(v,v',d)\in E$, we have $p_d = (s', S')$ with $v'\notin S'$.
    \item $Q(v) = \forall$ and there is $(v,v',d)\in E$ such that $p_d = (s', S')$ with $v'\notin S'$.
    \end{itemize}
\end{enumerate}
    
We claim that $L_\infty(A) = \pi_1(Y)$. 

Given a configuration $y\in Y$, we prove inductively the following statement: there is an accepting run starting from $(i,v)$ on $\pi_1(y)$  for all $i\in \Z^2$ and $v\in \pi_2(y_i)$. If $Q(v) = \exists$, then Condition 3 ensures we find an edge $(v, v', d)$ with $v'\in \pi_2(y_{i+d})$. Condition 1 and iterating this argument yields an accepting run starting from $(i+d, v')$. A similar argument works for $\forall$. Condition $2$ ensures that there is an initial state in $\pi_2(y_i)$, so $\pi_1(y)$ is accepted by $A$.

Conversely, given a configuration $x \in L_\infty(A)$, we define $y_i = (x_i, S_i)$ where $S_i \subset V$ is the set of states $v$ such that there is an accepting run on $x$ from $(i, v)$. Since $x$ admits an accepting run from any position for some initial state, it is easy to see that all three conditions above are satisfied and $y\in Y$.
\end{proof}

The proof of the following result is due to Ville Salo (personal communication).

\begin{theorem}\label{thm:soficAlt}
$\Alt \subsetneq \Sofic$.
\end{theorem}

\begin{definition}\label{def:lift}
Given a one-dimensional subshift $X\subset \Sigma^\Z$, its two-dimensional \emph{lift} is defined as 
\[X^\uparrow = \{y \in \Sigma^{\Z^2} : \exists x\in X, \forall i,j, y_{i,j} = x_i\}.\]
\end{definition}

By the constructions of Aubrun-Sablik \cite{aubrun2013simulation} or Durand-Romaschenko-Shen \cite{durand2012fixed}, the lift of any one-dimensional subshift given by a computable set of forbidden patterns is a sofic subshift. 

\begin{lemma}
If $X^\uparrow$ is accepted by an alternating PWA, then $X$ is sofic.
\end{lemma}

\begin{proof}
Let $A$ be the automaton that accepts $X^\uparrow$ in the sense that $X^\uparrow = L_\infty(A)$, and $A'$ be the automaton obtained by replacing every direction $\uparrow$ or $\downarrow$ by $\bullet$ in $A$. 

Since every configuration in $X^\uparrow$ is constant in the vertical direction, $A'$ accepts every configuration in $X^\uparrow$ (as well as additional configurations that are not constant vertically). Since $A'$ only travels horizontally, it can be seen as a two-way one-dimensional automaton that accepts exactly the one-dimensional configurations that lift to $X^\uparrow$; in other words, $A'$ accepts $X$. It follows that $X$ is defined by a regular set of forbidden patterns, so it is sofic.\end{proof}

To prove Theorem~\ref{thm:soficAlt}, see that the lift $X^\uparrow$ of any non-sofic one-dimensional subshift $X$ given by a computable set of forbidden patterns is sofic and not accepted by any alternating PWA. It is well-known that such subshifts exist; see e.g. \cite[Example 3.1.7]{LM}.

\section{Alternating Hierarchy for Plane-walking Automata}

In this paper, we are interested in comparing the power of plane-walking automata with varying access to nondeterminism. We introduce an alternating hierarchy of subshifts, similar to the classic alternating hierarchies of propositional logic formul\ae. We are not able to prove that this forms a strict hierarchy of subshifts, but we show that the hierarchy does not collapse at the first level.

\subsection{Definitions}

Starting from $\Delta_1$ (deterministic) PWA, we define $\Sigma_1$, resp. $\Pi_1$, PWA as the existential, resp. universal, PWA, that is, all states are labelled by $\exists$ quantifiers, resp. $\forall$ quantifiers. We call $\Sigma_1$ and $\Pi_1$ the corresponding classes of subshifts.

By definition, $\Delta_1 \subset \Sigma_1 \cap \Pi_1$. The rest of this paper is dedicated to show that $\Sigma_1$ and $\Pi_1$ are incomparable sets (and thus strictly larger than $\Delta_1$). First, we define the whole hierarchy inductively by using automata decompositions.

\begin{definition}[Decomposition] A two-dimensional automaton $A = (V,E,\Sigma, D, I, Q)$ has a \emph{decomposition} $(S, V \setminus S)$ for $S \subseteq V$ if there exists no path from $V \setminus S$ to $S$ in $E$. By extension, we call decomposition the pair of induced subautomata.
\end{definition}

\begin{definition}[$\Sigma_n$ and $\Pi_n$]
An alternating plane-walking automaton is $\Sigma_{n+1}$ if it admits a decomposition into a $\Sigma_1$-automaton and a $\Pi_{n}$-automaton. We denote $\Sigma_{n+1}$ the class of subshifts $X$ such that $X = L_\infty(A)$ for a $\Sigma_{n+1}$-automaton.

The definitions for $\Pi_{n+1}$ are symmetrical.
\end{definition}

Equivalently, a $\Sigma_n$ automata is such that the image by $Q$ of any path starting from $I$ is a word of $\exists^\ast\forall^\ast\dots$ with $n$ blocks ($n-1$ alternations). This justifies the following definition:

\begin{definition}[$\Delta_n$]
An alternating plane-walking automaton is $\Delta_{n}$ if the image by $Q$ of any path starting from $I$ is a word of $\{\exists,\forall\}^\ast$ with $n-1$ blocks ($n-2$ alternations). We denote by $\Delta_{n}$ the class of subshifts $X$ such that $X = L_\infty(A)$ for some $\Delta_n$ PWA $A$.
\end{definition}

Is is not difficult to see that $\Delta_n \subseteq \Sigma_n\subseteq \Delta_{n+1}$ and $\Delta_n\subseteq \Pi_n\subseteq \Delta_{n+1}$. We do not know whether $\Sigma_n \cap \Pi_n= \Delta_n$.

We are ready to state our main result:
  
\begin{theorem}\label{thm:main}
  $\Sigma_1 \not\subset \Pi_1$, and $\Pi_1 \not\subset \Sigma_1$. As a consequence, $\Delta_1 \subsetneq \Sigma_1$ and $\Delta_1 \subsetneq \Pi_1$.
\end{theorem}

In the next two subsections, we build two subshifts in $\Sigma_1 \backslash \Pi_1$ and $\Pi_1 \backslash \Sigma_1$, respectively.
We begin by a technical lemma:

\begin{lemma}
\label{lem:inter-SFT}
Let $X \in \Sigma_n$ be a subshift and $Y$ be a SFT. Then $X\cap Y\in\Sigma_n$. The same holds for $\Pi$ and $\Delta$.
\end{lemma}

\begin{proof}
Given a PWA $A$, define $A'$ as follows. From the initial state, check the area $\{\bullet, \uparrow, \rightarrow\}$ around the initial position. If a forbidden pattern is found, reject; this is done deterministically. Otherwise, come back to the initial position and execute a run of $A$. $A'$ belongs to the same class as $A$ and accepts all configurations in $L_\infty(A)$ where no forbidden pattern appears.
\end{proof}
  
\subsection{Sunny Side Up}

  \begin{definition}[Sunny side up]
    The sunny side up subshift $X_{ssu}$ is the set of configurations $x\in \{0, 1\}^{\Z^2}$ with at most one $i\in\Z^2$ such that $x_i = 1$.
  \end{definition}
  
  The sunny side up subshift can be easily accepted using a $\forall$-automaton that has the ability of exploring an unbounded space and rejecting if any "problem" is found in any of the branches.

  \begin{proposition} \label{prop-oui-pi}
    $X_{ssu} \in \Pi_1$.
  \end{proposition}

  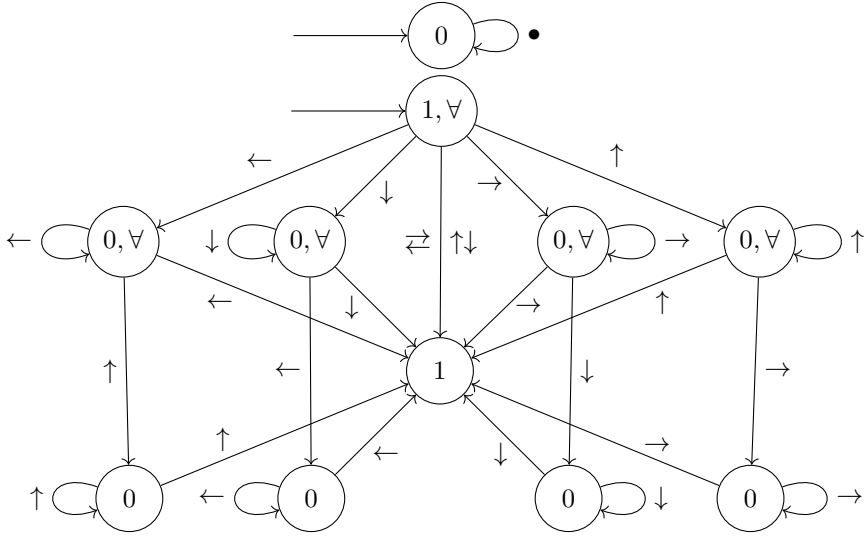
\begin{figure}[t]
    \centering
    \begin{tikzpicture}  [->, node distance = 1.5cm]
	\node[state] (a) {$0$};
	\node[state] (b) at (0, -1) {$1, \forall$};
	\node[state] (c) [below left =of b] {$0, \forall$};
	\node[state] (d) [left =of c] {$0, \forall$};
	\node[state] (e) [below right =of b] {$0, \forall$};
	\node[state] (f) [right =of e] {$0, \forall$};
    \node[state] (k) [below right =of c] {$1$};
    \node[state] (h) [below left =of k] {$0$};
	\node[state] (i) [below right =of k] {$0$};
	\node[state] (g) [left =of h] {$0$};
	\node[state] (j) [right =of i] {$0$};
    \node (input0) [left =of a] {};
    \node (input1) [left =of b] {};
	\path
      (input0) edge (a)
      (input1) edge (b)
      (b) edge node [above left] {$\leftarrow$} (d)
      (b) edge node [below right = -.1cm] {$\downarrow$} (c)
      (b) edge node [below left = -.1cm] {$\rightarrow$} (e)
      (b) edge node [above right] {$\uparrow$} (f)
      (d) edge node [left] {$\uparrow$} (g)
      (c) edge node [left] {$\leftarrow$} (h)
      (e) edge node [right] {$\downarrow$} (i)
      (f) edge node [right] {$\rightarrow$} (j)
      (b) edge node [right] {$\uparrow\downarrow$} node [left] {$\begin{matrix}\rightarrow\\[-8pt] \leftarrow\end{matrix}$} (k)
      (d) edge node [below, near start, yshift=-0.1cm] {$\leftarrow$} (k)
      (c) edge node [left, xshift=-0.1cm] {$\downarrow$} (k)
      (e) edge node [right] {$\rightarrow$} (k)
      (f) edge node [below, near start] {$\uparrow$} (k)
      (g) edge node [near start, above] {$\uparrow$} (k)
      (h) edge node [near start, right, xshift=0.1cm] {$\leftarrow$} (k)
      (i) edge node [near start, left, xshift=-0.1cm] {$\downarrow$} (k)
      (j) edge node [near start, above] {$\rightarrow$} (k)
      ;

	 \draw (a) to [out=20, in=-20, looseness=7] node [right] {$\bullet$} (a);
	 \draw (d) to [out=160, in=200, looseness=7] node [left] {$\leftarrow$} (d);
	 \draw (c) to [out=160, in=200, looseness=7] node [left] {$\downarrow$} (c);
	 \draw (e) to [out=20, in=-20, looseness=7] node [right] {$\rightarrow$} (e);
	 \draw (f) to [out=20, in=-20, looseness=7] node [right] {$\uparrow$} (f);
	 \draw (g) to [out=160, in=200, looseness=7] node [left] {$\uparrow$} (g);
	 \draw (h) to [out=160, in=200, looseness=7] node [left] {$\leftarrow$} (h);
	 \draw (i) to [out=20, in=-20, looseness=7] node [right] {$\downarrow$} (i);
	 \draw (j) to [out=20, in=-20, looseness=7] node [right] {$\rightarrow$} (j);

\end{tikzpicture}
    \caption{\label{fig-lssu-automata} A $\forall$-automaton accepting $X_{ssu}$. A branch visiting the central state has no next move, so it rejects.}
  \end{figure}  

\begin{proof}
Let $A$ be the automaton represented in Figure~\ref{fig-lssu-automata}. Every run starting from a position containing $0$ accepts. Therefore every configuration with no $1$ is accepted.
    
If $A$ starts on a $1$ at position $i$, it nondeterministically picks a quarter-plane to explore and rejects if this branch encounters another $1$. A run never visits a given position twice, so if $x$ contains a single $1$, all runs accept.

If $x$ contains two $1$ at positions $i$ and $j$, draw a path from $i$ to $j$ using at most two directions in $\{\rightarrow, \uparrow,\leftarrow,\downarrow\}$. This path yields a rejecting run starting from $i$. 
\end{proof}

\begin{proposition} \label{prop-non-sigma}
    $X_{ssu}\notin \Sigma_1$.
\end{proposition}

The following proof is related to Proposition~1 in \cite{planewalking}, which only holds for deterministic PWA (and a larger class of subshifts). Intuitively, $\Sigma_1$-automata are ill-fitted to explore unbounded spaces, as a run may reject from a given starting point only if all branches reject, but every branch may only visit a small part of the space. 

In the next proof, we use the notation $p\pm K = \{p\pm i : i\in K\}$ for $p\in\Z^2, K\subset \Z^2$.
\begin{proof}
Let $A$ be a $\Sigma_1$ PWA with $C$ states; we show that $X_{ssu} \neq L_\infty(A)$. For $\vec{v} \in \Z^2$ and $r\in \N$, we denote by $B^+(\vec{v},r)$ the set $\{j\in\Z^2 : \exists k\in \mathbb R^+, d(j,k\vec{v})\leq r\}$ (for the Manhattan distance on $\mathbb R^2$), and $B(\vec{v}, r) = B^+(\vec{v},r) \cup -B^+(\vec{v},r)$. In other words, this is the band of width $r$ around the (half-)line of direction $\vec{v}$ starting from $0$; if $\vec{v} = 0$, $B(0,r) = B^+(0,r)$ is a ball of radius $r$. When $r$ is not indicated, we assume that $r = C$.

Let $A$ be a $\Sigma_1$-automaton that accepts all configurations in $X_{ssu}$. Define $x \in X_{ssu}$ such that $x_i = 0$ for all $i$; $y \in X_{ssu}$ such that $y_0 = 1$ and $y_i = 0$ for all other positions $i$; and, for any $p\neq 0$, $z^p \notin X_{ssu}$ such that $z^p_0 = z^p_p = 1$ and $z^p_i = 0$ for all other positions $i$. We find some $p$ such that $A$ accepts $z^p$, showing that $A$ does not accept $X_{ssu}$.

To find some accepted $z^p$, we use the following property: if there is an accepting branch in $x$ that visits neither $0$ nor $p$, the same branch is accepting in $z^p$. Similarly, if an accepting branch in $y$ does not visit $p$ or an accepting branch in $\sigma^p(y)$ does not visit $0$, then the same branch is accepting in $z^p$.
\medskip

Let $S_1(r) \subset V\times B(0,r)$ be such that $(s,k)\in S_1$ if and only if there is an accepting run on $y$ from $(s,k)$. For each $(s,k)\in S_1$, pick an arbitrary accepting branch $b$. It is described by a sequence $(s_n, p_n)_{n\in\N} \in (V\times \Z^2)^\N$. If $b$ stays in $B(0)$, we find two steps $i<j$ such that $(s_i,p_i) = (s_j,p_j)$ and, by a pumping argument, we build another accepting branch $b_p$ which is eventually periodic and stays in $B(0)$. If $b$ leaves $B(0)$, we find two steps $i<j$ such that $s_i = s_j$ and $d(0,p_j) \geq d(0,p_i)$ and, by a pumping argument, we build another accepting branch $b_p$ which is eventually periodic up to translation by the pumping vector $v^{(s,k)} = p_j - p_i$ every $j-i$ steps; in other words, $b_p$ stays in some band $B(v^{(s,k)})$.
Denote $P_1(r) = \{v^{(s,k)} : (s,k)\in S_1(r)\}$ the set of all such pumping vectors.

Let $S_0\subset V$ be the set of states reachable from $i_0$ through a path labelled by $0$. Since $x$ is accepted, there must be a cycle that stays in $S_0$. For any such simple cycle (without repeated states), the associated pumping vector is the sum of all directions. Denote $P_0$ the finite set of all such pumping vectors.

We distinguish two cases that we illustrate in Figure~\ref{fig-ssupump}.

\paragraph*{All vectors in $P_0$ are colinear to some $v$.} We choose $p$ such that $p \notin B(v) \cup \bigcup_{v_1 \in P_1(0)} B(v_1)$. This avoids a finite set of one-dimensional subsets, so such a choice is possible.

Take any position $k\in\Z^2$ and assume that $k\notin p-B^+(v)$. Pick an accepting branch $b$ in $y$ from $k$. As long as $b$ never visits the $1$ in position $0$, it stays in a state of $S_0$, so we can pump to build a periodic accepting branch that stays in $k+B^+(v)$ and does not visit $p$. If $b$ visits $0$, then it is in some state $s$ such that $(s,0)\in S_1(0)$, and so we can pump to build another accepting branch $b_p$ that stays in $B(0)$ or in $B^+(v_1)$ for some $v_1\in P_1(0)$. Either way, we built an accepting branch in $y$ from $k$ that does not visit $p$, so it is accepting in $z^p$.

If $k\in p-B^+(v)$, then $k\notin 0-B^+(v)$ since we chose $p$ in such a way that the sets are disjoint. With a similar argument, we build an accepting branch that does not visit $0$ on $\sigma^p(y)$.

Every starting position in $z^p$ has an accepting branch, so $z^p$ is accepted by $A$.

\paragraph*{There are two noncolinear vectors $v_a, v_b$ in $P_0$.} This means that, in $x$, the accepting run from $(i_0, 0)$ has two accepting branches that stay in $B^+(v_a)$ and $B^+(v_b)$, respectively. Since $v_a$ and $v_b$ are not colinear, there is $r>0$ large enough that $B(v_a)\cap B(v_b) \subset B(0,r)$. 

We choose $p$ such that $p \notin \bigcup_{v_1 \in P_1} B(v_1, r+C)$ and such that $p$ is in the quarterplane generated by $v_a, v_b$, at distance at least $r+2C$ from the border.

Take a position $k\in\Z^2$. If $k \in B(0,r)$, then we pump to build an accepting branch $b_p$ in $y$ that stays either in $B(0)$ or in $k+B^+(v_1)$ for some $v_1\in P_1(r)$, so that $b_p$ does not visit $p$. If $k \in B(p,r)$, the same argument applies on $\sigma^p(y)$. 

If $k \notin B(0,r)\cup B(p,r)$, then one of the two bands $B_a = k+B^+(v_a)$ and $B_b = k+B^+(v_b)$ contains neither $0$ nor $p$. Indeed, 
\begin{itemize}
\item $B_a\cap B_b$ contains neither position by definition of $r$;
\item If both positions appeared, we would have $|p - 0 - \pm(av_a - bv_b)|\leq 2C$ for some $a,b\geq 0$, which contradicts the choice of $p$.
\end{itemize}
Assuming that it is $B_a$, we pump to build an accepting branch in $x$ from $k$ that stays in $B_a$ and visits neither $0$ nor $p$.

Every position in $z^p$ has an accepting branch, so $z^p$ is accepted by $A$.
\end{proof}

\begin{figure}[t]
    \centering
    \begin{minipage}{0.49\textwidth}
    \pgfmathsetmacro \gs{0.5}

\begin{tikzpicture}


    \draw[dashed] (-1 * \gs, 6 * \gs)
      -- (3 * \gs, 2 * \gs);
    \draw[dashed] (-1 * \gs, 8 * \gs)
      -- (4 * \gs, 3 * \gs);

    \draw[dashed] (2 * \gs, 11 * \gs)
      -- (6 * \gs, 7 * \gs);
    \draw[dashed] (4 * \gs, 11 * \gs)
      -- (7 * \gs, 8 * \gs);


  \draw[dashed] (3 * \gs, 3.67 * \gs)
    -- (11.0 * \gs, 3.67 * \gs);
  \draw[dashed] (3 * \gs, 2.33 * \gs)
    -- (11.0 * \gs, 2.33 * \gs);

  \draw[dashed] (7 * \gs, 7.67 * \gs)
    -- (11.0 * \gs, 7.67 * \gs);
  \draw[dashed] (7 * \gs, 6.33 * \gs)
    -- (11.0 * \gs, 6.33 * \gs);


  \draw (3.67 * \gs, 3 * \gs)
    -- (3.67 * \gs, -1 * \gs);
  \draw (2.33 * \gs, 3 * \gs)
    -- (2.33 * \gs, -1 * \gs);

  \draw (7.67 * \gs, 7 * \gs)
    -- (7.67 * \gs, -1 * \gs);
  \draw (6.33 * \gs, 7 * \gs)
    -- (6.33 * \gs, -1 * \gs);


  \foreach \x\y\dx\dy\la in {
    5/5/0/1/$v$,
    9.5/0/0/1/$v$,
    0/0/0/1/$v$
  } {
    \node[circle, fill=black, minimum size = 2, inner sep=0]
      at (\x * \gs, \y * \gs) {};
    \draw[->] (\x * \gs, \y * \gs) --
      node[right] {\la}
      (\x * \gs + \dx * \gs, \y * \gs + \dy * \gs);
  }


  \foreach \x\y\pos in {3/3/0, 7/7/p} {
    \node[circle, draw,fill=white, minimum size=60*\gs] at (\x * \gs, \y * \gs) {};
    \node at (\x * \gs, \y * \gs) {$1$};
  }


  \foreach \x\y\dx\dy\la in {
    2.5/3.3/-1/1/$v_1$,
    0.5/5.4/-1/1/$v_1$,
    6.5/7.4/-1/1/$v_1$, 
    4.5/9.4/-1/1/$v_1$
  } {
    \node[circle, fill=black, minimum size = 2, inner sep=0]
      at (\x * \gs, \y * \gs) {};
    \draw[->] (\x * \gs, \y * \gs) --
      node[above, xshift=0.075cm] {\la}
      (\x * \gs + \dx * \gs, \y * \gs + \dy * \gs);
  }

  \foreach \x\y\dx\dy\la in {
    3.8/2.7/1.4/0/$v'_1$,
    8.8/2.7/1.4/0/$v'_1$,
    7.8/6.7/1.4/0/$v'_1$
  } {
    \node[circle, fill=black, minimum size = 2, inner sep=0]
      at (\x * \gs, \y * \gs) {};
    \draw[->] (\x * \gs, \y * \gs) --
      node[above, yshift=-0.035cm] {\la}
      (\x * \gs + \dx * \gs, \y * \gs + \dy * \gs);
  }

  \draw[->] (3.2 * \gs, -1 * \gs)
    -- (3.2 * \gs, 1 * \gs)
    -- (2.6 * \gs, 1.8 * \gs)
    -- (2.6 * \gs, 2.5 * \gs);

  \draw[->] (2.8 * \gs, -1 * \gs)
    -- (2.8 * \gs, 0 * \gs)
    -- (3.2 * \gs, 1.8 * \gs)
    -- (3.2 * \gs, 2.5 * \gs);

  \draw[->] (6.8 * \gs, -1 * \gs)
    -- (6.8 * \gs, 2 * \gs)
    -- (7.2 * \gs, 3.4 * \gs)
    -- (7.2 * \gs, 5 * \gs)
    -- (6.6 * \gs, 5.8 * \gs)
    -- (6.6 * \gs, 6.5 * \gs);

  \draw[->] (7.2 * \gs, -1 * \gs)
    -- (7.2 * \gs, 2 * \gs)
    -- (6.8 * \gs, 3 * \gs)
    -- (6.8 * \gs, 4 * \gs)
    -- (7.2 * \gs, 5.8 * \gs)
    -- (7.2 * \gs, 6.5 * \gs);

\end{tikzpicture}
    \end{minipage}
    \begin{minipage}{0.49\textwidth}
    \pgfmathsetmacro \gs{0.5}

\begin{tikzpicture}


    \draw[dashed] (-1 * \gs, 6 * \gs)
      -- (3 * \gs, 2 * \gs);
    \draw[dashed] (-1 * \gs, 8 * \gs)
      -- (4 * \gs, 3 * \gs);

    \draw[dashed] (2 * \gs, 11 * \gs)
      -- (6 * \gs, 7 * \gs);
    \draw[dashed] (4 * \gs, 11 * \gs)
      -- (7 * \gs, 8 * \gs);


  \draw[dashed] (3 * \gs, 3.67 * \gs)
    -- (11.0 * \gs, 3.67 * \gs);
  \draw[dashed] (3 * \gs, 2.33 * \gs)
    -- (11.0 * \gs, 2.33 * \gs);

  \draw[dashed] (7 * \gs, 7.67 * \gs)
    -- (11.0 * \gs, 7.67 * \gs);
  \draw[dashed] (7 * \gs, 6.33 * \gs)
    -- (11.0 * \gs, 6.33 * \gs);


  \foreach \x\y\dx\dy\la in {
    5/5/0/2/$v_a$,
    9.5/0/0/2/$v_a$,
    0/0/0/2/$v_a$
  } {
    \node[circle, fill=black, minimum size = 2, inner sep=0]
      at (\x * \gs, \y * \gs) {};
    \draw[->] (\x * \gs, \y * \gs) --
      node[right] {\la}
      (\x * \gs + \dx * \gs, \y * \gs + \dy * \gs);
  }


  \foreach \x\y\pos in {3/3/0, 7/7/p} {
    \node[circle, draw,fill=white, minimum size=60*\gs] at (\x * \gs, \y * \gs) {};
    \node at (\x * \gs, \y * \gs) {$1$};
  }


  \foreach \x\y\dx\dy\la in {
    2.5/3.3/-1/1/$v_1$,
    0.5/5.4/-1/1/$v_1$,
    6.5/7.4/-1/1/$v_1$, 
    4.5/9.4/-1/1/$v_1$
  } {
    \node[circle, fill=black, minimum size = 2, inner sep=0]
      at (\x * \gs, \y * \gs) {};
    \draw[->] (\x * \gs, \y * \gs) --
      node[above, xshift=0.075cm] {\la}
      (\x * \gs + \dx * \gs, \y * \gs + \dy * \gs);
  }

  \foreach \x\y\dx\dy\la in {
    3.8/2.7/1.4/0/$v'_1$,
    8.8/2.7/1.4/0/$v'_1$,
    7.8/6.7/1.4/0/$v'_1$
  } {
    \node[circle, fill=black, minimum size = 2, inner sep=0]
      at (\x * \gs, \y * \gs) {};
    \draw[->] (\x * \gs, \y * \gs) --
      node[above, yshift=-0.035cm] {\la}
      (\x * \gs + \dx * \gs, \y * \gs + \dy * \gs);
  }

  \foreach \x\y\dx\dy\la in {
    0/0/2/0/$v_b$,
    9.5/0/2/0/$v_b$,
    5/5/2/0/$v_b$ 
  } {
    \node[circle, fill=black, minimum size = 2, inner sep=0]
      at (\x * \gs, \y * \gs) {};
    \draw[->] (\x * \gs, \y * \gs) --
      node[below] {\la}
      (\x * \gs + \dx * \gs, \y * \gs + \dy * \gs);
  }

\end{tikzpicture}
    \end{minipage}
    \caption{\label{fig-ssupump}
        Left: the first case when all vectors in $P_0$ are colinear to $v$. Right: the second case when $\{v_a, v_b\}\subset P_0$. In both cases, $P_1 = \{v_1, v'_1\}$ and the circles represent positions $0$ and $p$. The proof shows that each initial position admits a path that never visits $0$ or $p$ and accepts in $y$ or $\sigma^p(y)$. 
      }
  \end{figure}
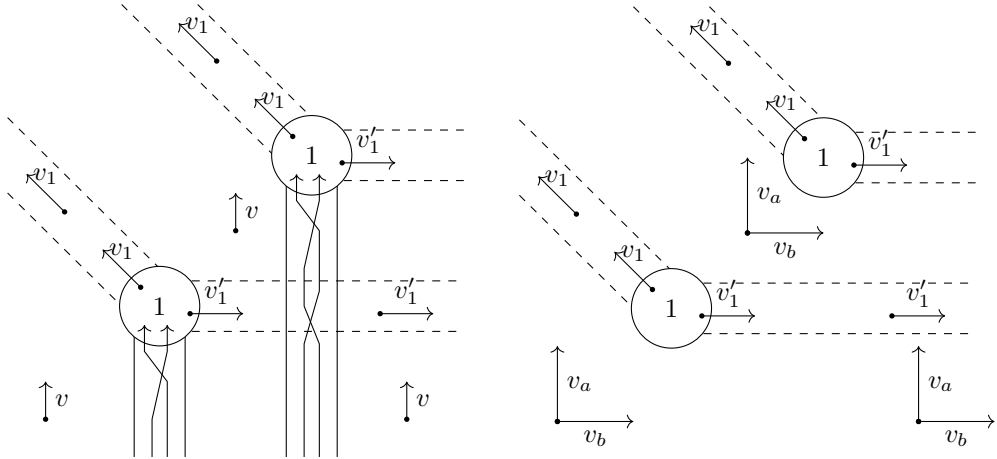

\subsection{The Cone Labyrinth}

\begin{definition} \label{def-cl}
    Let $\Sigma = \{0, 1, \#\}$.
    The cone labyrinth subshift (denoted $X_{cl}$) is the set of configurations $x$ which contains none of the forbidden patterns $010, 11, \#1\#, \vstack{0}{\#}, \vstack{\#}{0}, \vstack{1}{\#}, \vstack{\#}{1}$ and such that from any position with a pattern $\#1$, there is a path to a position with a pattern $1\#$ using steps $\nearrow, \rightarrow, \searrow$ that only visits positions with symbols $0$.
\end{definition}

In a configuration $x\in\{0,1,\#\}^{\Z^2}$ of $X_{cl}$, every column contains either only $\#$ symbols, or only $0$ and $1$ symbols. The $\#$ marks the outside areas, the $0$ the inside areas, and $1$ corresponds to entrances / exits to change areas. In particular, a $1$ can only be between a $0$ and a $\#$. Furthermore, every entrance $\#1$ can be matched to at least one exit $1\#$, in the sense that one can walk from $\#1$ to $1\#$ through zeroes using steps $\nearrow, \rightarrow, \searrow$.

In other words, if the width of the inside area is $k$, then every entrance must be matched to an exit at a vertical distance at most $k$.

\begin{proposition} \label{prop-oui-sigma}
  $X_{cl}\in \Sigma_1$.
\end{proposition}

\begin{proof}
We build a $\exists$ automaton that accepts $X_{cl}$ assuming that patterns from Definition~\ref{def-cl} do not appear; we can do this assumption by Lemma~\ref{lem:inter-SFT}. The automaton is represented in Figure~\ref{fig-cl-automata}.

  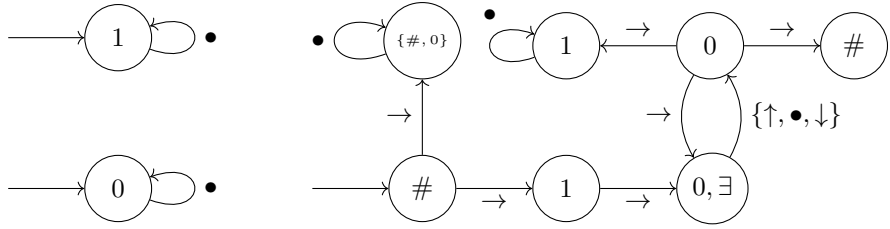
\begin{figure}[t]
    \centering
    \begin{tikzpicture}  [->]

    \node[state] (0) at (-4,0) {$0$};
    \node[state] (1) at (-4,2) {$1$};
    \node (in0) [left =of 0] {};
    \node (in1) [left =of 1] {};
    \path (in0) edge (0) (in1) edge (1);
	\draw (0) to [out=-20, in=20, looseness=7] node [right] {$\bullet$} (0);
    \draw (1) to [out=-20, in=20, looseness=7] node [right] {$\bullet$} (1);

	\node[state] (a) {$\#$};
    
    \node (input0) [left =of a] {};
    \node[state] (d) [above =of a] {\tiny $\{\#, 0\}$};
	\node[state] (b) [right =of a] {$1$};
    \node[state] (e) [above =of b] {$1$};

	\node[state] (c) [right =of b] {$0, \exists$};
	\node[state] (f) [right =of e] {$0$};
    \node[state] (g) [right =of f] {$\#$};

	\path
	  (input0) edge (a)
      (a) edge node [below] {$\rightarrow$} (b)
      (a) edge node [left] {$\rightarrow$} (d)
      (b) edge node [below] {$\rightarrow$} (c)
      (c) edge [bend right] node [right] {$\uparrow \bullet \downarrow$} (f)
      (f) edge [bend right] node [left] {$\rightarrow$} (c)
      (f) edge node [above] {$\rightarrow$} (e)
      (f) edge node [above] {$\rightarrow$} (g)
      (c) edge [bend left=10] node [left] {$\uparrow \downarrow$} (e)
      ; 	
   \draw (d) to [out=200, in=160, looseness=7] node [left] {$\bullet$} (d);
	 \draw (e) to [out=200, in=160, looseness=7] node [above = .20cm] {$\bullet$} (e);

\end{tikzpicture}
    \caption{\label{fig-cl-automata} $\exists$-automaton accepting $X_{cl}$. We assumed for readability that the patterns from Definition~\ref{def-cl} have already been forbidden.}
  \end{figure}

If the initial position is not the entrance of a labyrinth $\#1$, accept (by looping). Otherwise, walk nondeterministically in all directions $\nearrow, \rightarrow, \searrow$. Keep going as long as you see $0$, accept if you find a $1$, reject if you find a $\#$. Therefore the automata accepts if and only if, starting from every entrance, one branch found a matching exit.
\end{proof}

\begin{proposition} \label{prop-non-pi}
  $X_{cl} \notin \Pi_1$.
\end{proposition}

Intuitively, in order for a run to reject a labyrinth with no exit, some individual branch must reject, which requires visiting a region of unbounded size (depending on the width on the labyrinth). By making the width large enough, we disorient this run into rejecting a valid configuration because it cannot check all required cells.

\begin{proof}
Let $A$ be a $\Pi_1$-automaton that rejects all configurations not in $X_{cl}$. We build a configuration $y\in X_{cl}$ that $A$ rejects, the process being illustrated in Figure~\ref{fig-labypump}.

First define a configuration $x^n\notin X_{cl}$ for some $n>|Q|$.  
\[x^n_{(i,j)} = 
\left\{\begin{array}{cl}
1&\text{if }(i,j) = (0,0)\\
0&\text{otherwise, if }0\leq i\leq n\\
\#&\text{otherwise}
\end{array}\right.\] 
Since $A$ rejects $x^n$, there exists a position from which there is no accepting run of $A$; that is, some branch $b$ from some position $p$ rejects (in finite time). The configuration would be valid if we switched the symbols at positions $(0,0)$ or $(n, k)$ for any $-n\leq k\leq n$; therefore $b$ must visit all these positions, since it would otherwise reject a valid configuration.

Within $x^n$, we call "the left" the half-plane $i\leq 0$, "the right" the half-plane $i\geq n$, and "the center" the band $0<i<n$. For simplicity, we assume that $b$ starts in the left and ends in the right. We decompose $b$ as $\ell_0c^{\ell\to r}_0r_0c^{r\to \ell}_0\dots \ell_Tc^{\ell\to r}_Tr_T$, where for all $k = 0, \dots, T$:
\begin{itemize}
\item every subsequence is nonempty; 
\item $\ell_k$ starts in the left, stays in the left and center, and ends in the left (\emph{left stays});
\item $r_k$ starts in the right, stays in the right and center, and ends in the right (\emph{right stays});
\item $c^{\ell\to r}_k$ and $c^{r\to\ell}_k$ stay in the center (\emph{center crossings}).
\end{itemize}

Consider the first crossing $c^{\ell\to r}_0$. Crossing takes at least $n$ steps, so we find two positions $(i,j)$ and $(i+a_0,j+b_0)$ with $a_0 > 0$ that $c^{\ell\to r}_0$ visited in that order in the same state; $(a_0, b_0)$ is the pumping vector. Similarly, we find such pumping vectors $(a_k, b_k)$ with $a_k>0$ for all $c^{\ell\to r}_k$ and $(\alpha_k, \beta_k)$ with $\alpha_k<0$ for all $c^{r\to \ell}_k$.

By pumping on these vectors on each crossing, we build a branch $b'$ that will be valid on some other configuration $y\neq x^n$; for the time being, we only pay attention to the positions in the branch and not to the underlying configuration.

Let $C$ be the lowest common multiple of all $a_k$ and $-\alpha_k$ and let $m>0$ to be fixed later. We define $b' = \ell'_0c'^{\ell\to r}_0r'_0c'^{r\to \ell}_0\dots \ell'_Tc'^{\ell\to r}_Tr'_T$ step by step.

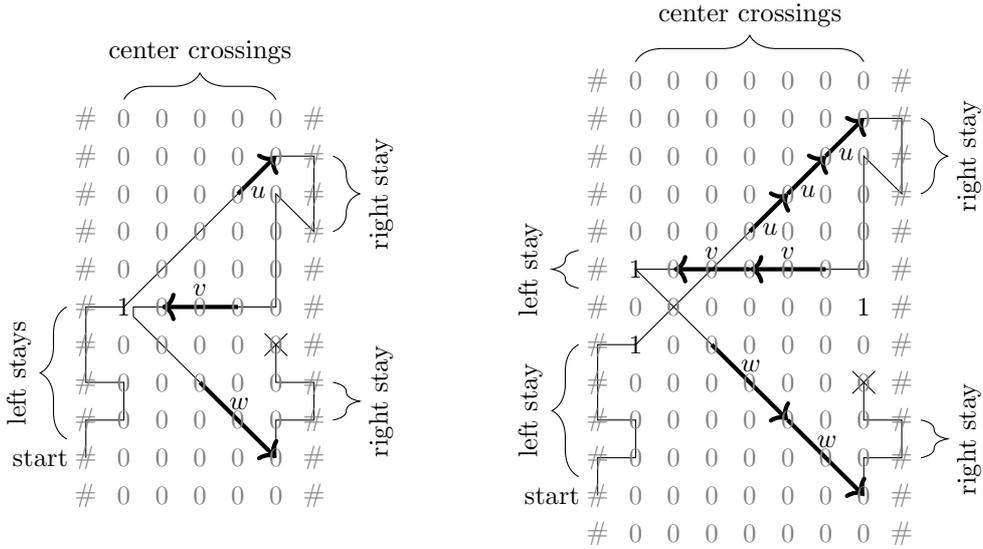
\begin{figure}[t]
    \centering
    \begin{minipage}{0.49\textwidth}
    \pgfmathsetmacro \gs{0.5}

\begin{tikzpicture}


  \node at (-1.2 * \gs, 1 * \gs) {start};

  \draw (0 * \gs, 1 * \gs)
    -- (0 * \gs, 2 * \gs)
    -- (1 * \gs, 2 * \gs)
    -- (1 * \gs, 3 * \gs)
    -- (0 * \gs, 3 * \gs)
    -- (0 * \gs, 5 * \gs)
    -- (1 * \gs, 5 * \gs)
    -- (5 * \gs, 9 * \gs)
    -- (6 * \gs, 9 * \gs)
    -- (6 * \gs, 7 * \gs)
    -- (5 * \gs, 8 * \gs)
    -- (5 * \gs, 5 * \gs)
    -- (4 * \gs, 5 * \gs);
  \draw[->, ultra thick] (4 * \gs, 5 * \gs)
    -- node[above] {$v$} (2 * \gs, 5 * \gs);
  \draw (2 * \gs, 5 * \gs)
    -- (1.25 * \gs, 5 * \gs)
    -- (1.25 * \gs, 4.75 * \gs)
    -- (3 * \gs, 3 * \gs);
  \draw[->, ultra thick] (3 * \gs, 3 * \gs)
    -- node[above, xshift=0.5] {$w$} (5 * \gs, 1 * \gs);
  \draw (5 * \gs, 1 * \gs)
    -- (5 * \gs, 2 * \gs)
    -- (6 * \gs, 2 * \gs)
    -- (6 * \gs, 3 * \gs)
    -- (5 * \gs, 3 * \gs)
    -- (5 * \gs, 4 * \gs);

  \draw[->, ultra thick] (4 * \gs, 8 * \gs)
    -- node[below, xshift=0.5] {$u$} (5 * \gs, 9 * \gs);

  \draw (4.7 * \gs, 3.7 * \gs) -- (5.3 * \gs, 4.3 * \gs);
  \draw (5.3 * \gs, 3.7 * \gs) -- (4.7 * \gs, 4.3 * \gs);

  \draw [decorate,decoration={brace,amplitude=10pt}]
    (-0.5 * \gs, 1.5 * \gs) -- (-0.5 * \gs, 5 * \gs)
    node [midway,above,rotate=90,yshift=10] {left stays};

  \draw [decorate,decoration={brace,amplitude=10pt}]
    (1 * \gs, 10.5 * \gs) -- (5 * \gs, 10.5 * \gs)
    node [midway,above,yshift=10] {center crossings};

  \draw [decorate,decoration={mirror,brace,amplitude=10pt}]
    (6.5 * \gs, 7 * \gs) -- (6.5 * \gs, 9 * \gs)
    node [midway,below,rotate=90,yshift=-10] {right stay};

  \draw [decorate,decoration={mirror,brace,amplitude=10pt}]
    (6.5 * \gs, 2 * \gs) -- (6.5 * \gs, 3 * \gs)
    node [midway,below,rotate=90,yshift=-10] {right stay};


  \foreach \y in {0, ..., 10} {
    \foreach \x in {0, 6} {
      \node[gray] at (\x * \gs, \y * \gs) {$\#$};
    }
  }

  \foreach \y in {0, ..., 10} {
    \foreach \x in {2, 3, 4, 5} {
      \node[gray] at (\x * \gs, \y * \gs) {$0$};
    }
  }

  \foreach \y in {0, ..., 4, 6, 7, 8, 9, 10} {
    \node[gray] at (1 * \gs, \y * \gs) {$0$};
  }
  
  \node at (1 * \gs, 5 * \gs) {$1$};
\end{tikzpicture}
    \end{minipage}
    \begin{minipage}{0.49\textwidth}
    \pgfmathsetmacro \gs{0.5}

\begin{tikzpicture}


  \node at (-1.2 * \gs, 1 * \gs) {start};

  \draw (0 * \gs, 1 * \gs)
    -- (0 * \gs, 2 * \gs)
    -- (1 * \gs, 2 * \gs)
    -- (1 * \gs, 3 * \gs)
    -- (0 * \gs, 3 * \gs)
    -- (0 * \gs, 5 * \gs)
    -- (1 * \gs, 5 * \gs)
    -- (7 * \gs, 11 * \gs)
    -- (8 * \gs, 11 * \gs)
    -- (8 * \gs, 9 * \gs)
    -- (7 * \gs, 10 * \gs)
    -- (7 * \gs, 7 * \gs)
    -- (1 * \gs, 7 * \gs)
    -- (7 * \gs, 1 * \gs)
    -- (7 * \gs, 2 * \gs)
    -- (8 * \gs, 2 * \gs)
    -- (8 * \gs, 3 * \gs)
    -- (7 * \gs, 3 * \gs)
    -- (7 * \gs, 4 * \gs);

  \draw[->, ultra thick] (4 * \gs, 8 * \gs)
   -- node[below, xshift=0.5] {$u$} (5 * \gs, 9 * \gs);
  \draw[->, ultra thick] (5 * \gs, 9 * \gs)
   -- node[below, xshift=0.5] {$u$} (6 * \gs, 10 * \gs);
  \draw[->, ultra thick] (6 * \gs, 10 * \gs)
   -- node[below, xshift=0.5] {$u$} (7 * \gs, 11 * \gs);

  \draw[->, ultra thick] (6 * \gs, 7 * \gs)
    -- node[above] {$v$} (4 * \gs, 7 * \gs);
  \draw[->, ultra thick] (4 * \gs, 7 * \gs)
    -- node[above] {$v$} (2 * \gs, 7 * \gs);

  \draw[->, ultra thick] (3 * \gs, 5 * \gs)
    -- node[above, xshift=0.5] {$w$} (5 * \gs, 3 * \gs);
  \draw[->, ultra thick] (5 * \gs, 3 * \gs)
    -- node[above, xshift=0.5] {$w$} (7 * \gs, 1 * \gs);

  \draw (6.7 * \gs, 3.7 * \gs) -- (7.3 * \gs, 4.3 * \gs);
  \draw (7.3 * \gs, 3.7 * \gs) -- (6.7 * \gs, 4.3 * \gs);

  \draw [decorate,decoration={brace,amplitude=10pt}]
    (-0.5 * \gs, 1.5 * \gs) -- (-0.5 * \gs, 5 * \gs)
    node [midway,above,rotate=90,yshift=10] {left stay};

  \draw [decorate,decoration={brace,amplitude=10pt}]
    (-0.5 * \gs, 6.5 * \gs) -- (-0.5 * \gs, 7.5 * \gs)
    node [midway,above,rotate=90,yshift=10] {left stay};

  \draw [decorate,decoration={brace,amplitude=10pt}]
    (1 * \gs, 12.5 * \gs) -- (7 * \gs, 12.5 * \gs)
    node [midway,above,yshift=10] {center crossings};

  \draw [decorate,decoration={mirror,brace,amplitude=10pt}]
    (8.5 * \gs, 9 * \gs) -- (8.5 * \gs, 11 * \gs)
    node [midway,below,rotate=90,yshift=-10] {right stay};

  \draw [decorate,decoration={mirror,brace,amplitude=10pt}]
    (8.5 * \gs, 2 * \gs) -- (8.5 * \gs, 3 * \gs)
    node [midway,below,rotate=90,yshift=-10] {right stay};


  \foreach \y in {0, ..., 12} {
    \foreach \x in {0, 8} {
      \node[gray] at (\x * \gs, \y * \gs) {$\#$};
    }
  }

  \foreach \y in {0, ..., 12} {
    \foreach \x in {2, ..., 6} {
      \node[gray] at (\x * \gs, \y * \gs) {$0$};
    }
  }

  \foreach \y in {0, ..., 4} {
    \node[gray] at (1 * \gs, \y * \gs) {$0$};
  }

  \node at (1 * \gs, 5 * \gs) {$1$};
  \node[gray] at (1 * \gs, 6 * \gs) {$0$};
  \node at (1 * \gs, 7 * \gs) {$1$};

  \foreach \y in {8, ..., 12} {
    \node[gray] at (1 * \gs, \y * \gs) {$0$};
  }

  \foreach \y in {0, ..., 5} {
    \node[gray] at (7 * \gs, \y * \gs) {$0$};
  }
  
  \node at (7 * \gs, 6 * \gs) {$1$};

  \foreach \y in {7, ..., 12} {
    \node[gray] at (7 * \gs, \y * \gs) {$0$};
  }
\end{tikzpicture}
    \end{minipage}
    \caption{\label{fig-labypump}
        On the left, a rejecting run over the invalid labyrinth $x^4$, which has no exit.
        Bolded parts $\vec{u}$, $\vec{v}$ and $\vec{w}$ represent pumping vectors which start and end in the same state.        
        These vectors are repeated on the right to provide another rejecting run of the same automata over the valid labyrinth $y$.
      }
      
  \end{figure}

\begin{itemize}
\item $\ell'_0 = \ell_0$ is unchanged;
\item $c'^{\ell\to r}_0$ is $c^{\ell\to r}_0$ that we pump $\frac{mC}{a_0}$ times along the vector $(a_0, b_0)$.
\item $r'_0 = \sigma^{(mC, \frac{mC}{a_1} b_1)}(r_0)$; in other words, $r'_0$ is shifted relative to $r_0$ so that the positions are consistent with the previous part.
\item $c'^{r\to \ell}_0$ is $c^{r\to \ell}_0$ shifted by $(mC, \frac{mC}{a_1} b_1)$ and pumped $\frac{mC}{-\alpha_0}$ times along the vector $(\alpha_0, \beta_0)$.
\end{itemize}

In a similar manner, we pump every crossing in the center and shift: 

\begin{itemize}
\item $\ell'_k$ and $c'^{\ell\to r}_k$ by $(0, o^\ell_k(m))$, where $o^\ell_k(m) = \Sigma_{i = 0}^{k-1} \frac{mC}{a_i} b_i + \Sigma_{i = 0}^{k-1} \frac{mC}{-\alpha_i} \beta_i$;
\item $r'_k$ and $c'^{r\to \ell}_k$ by $(mC, o^r_k(m))$ where $o^r_k(m) = \Sigma_{i = 0}^{k} \frac{mC}{a_i} b_i + \Sigma_{i = 0}^{k-1} \frac{mC}{-\alpha_i} \beta_i$.
\end{itemize}
These values are chosen so that the positions are locally consistent with allowed transitions in $A$. Notice that $o^\ell_k(m) - o^\ell_{k'}(m)$ is either always $0$ or tends to $\pm\infty$ as $m\to \infty$. 

Denoting $\varphi$ the footprint, we see that $\varphi(\ell'_k) = \varphi(\ell_k) + o^\ell_k(m)$. Because every $\varphi(\ell_k)$ is finite, we can find $m$ large enough such that, for all $k$ and $k'$, $\varphi(\ell'_k) \cap \varphi(\ell'_{k'}) = \emptyset$ or $\varphi(\ell'_k) \cap \varphi(\ell'_{k'}) = \varphi(\ell_k) \cap \varphi(\ell_{k'})$, and similarly for right stays. In other words, two stays either visit no cell in common or they cross in exactly the same manner as the original branch.

If needed, we increase $m$ so that $m>|\varphi(b)|$.

Now we build a configuration $y$ so that $b'$ is a branch of the run starting at the same position $p$. Start by putting $y = x^{n + mC}$, and do the following modifications: 

\begin{enumerate}
\item set $y_{(0, o^\ell_k(m))} = 1$ for all $k$;
\item choose some $i$ such that $|i|<n+Cm$ and $(n+Cm, i) \notin \bigcup_k\varphi(r'_k)$. Then set $y_{(n,i+o^\ell_k(m))} = 1$ for all $k$. 
\end{enumerate}

Condition 1 adds entrances at such positions that the left stay $\ell'_k$ "sees" an entrance exactly when the original left stay $\ell_k$ saw an entrance at position $(0,0)$. This does not impact other left stays by the argument above.

Condition 2 adds exits at positions that are not visited by $b'$, which is possible because we chose $m$ to be larger than the footprint of $b$, but satisfy the definition for a valid labyrinth. This ensures that $y\in X_{cl}$ and that $b'$ rejects.

By construction, the branch $b'$ is a branch for the run on $y$ starting at $p$. It follows that $y\in X_{cl}$ is rejected by $A$, a contradiction.
\end{proof}

\section{Conclusion}

\subsection{Summary and open questions}
We proved that subshifts accepted by alternating plane-walking automata form a strict subset of sofic subshifts, and that the first level of the hierarchy with bounded alternations is strict: $\Sigma_1$ and $\Pi_1$ are incomparable, and thus the inclusion of $\Delta_1$ in both of them is strict.

We sum up our open questions:
\begin{enumerate}
\item Is the hierarchy strict for all $n$ or does it collapse at some level? For example, can we find a subshift in $\Sigma_{2} \backslash \Sigma_{1}$?
\item If a subshift is accepted by a universal PWA and an existential PWA, is it also accepted by a deterministic PWA? More generally, is it the case that $\Delta_n = \Sigma_n \cap \Pi_n$? 
\item Is there an equivalent definition for subshifts accepted by alternating plane-walking automata by forbidden patterns accepted by plane-walking automata that do not require taking a complement?
\end{enumerate}

Pumping arguments are tedious even in the first floors of the hierarchy, and we would like to find other tools, for example (as suggested by Guillaume Theyssier and inspired by \cite{terrier19}) from communication complexity.

This is only one of multiple possible hierarchies on walking automata. We mention $n$-nested automata in the next section. Automata with the ability to leave up to $n$ pebbles (non-moving marks used as memory) during a run have also been considered. Salo and Törma studied automata with multiple heads and this hierarchy collapses at $n=3$ \cite{planewalking,salo2020four}.

Definitions~\ref{def:recognising} and \ref{def:PWA} (recognised / accepted subshifts) extend to higher dimensions $\mathbb Z^d, d\geq 2,$ or any finitely generated group by replacing the set of directions by $S$ and $S\cup S^{-1}\cup\{\bullet\}$, respectively, where $S$ is a set of generators of the group. While our main results extend directly to $\mathbb Z^d$ by considering lifts of our two-dimensional examples (see Definition~\ref{def:lift}), we make no guesses as to the situation in more complicated groups.

\subsection{Strict Hierarchy and Tree-walking Automata}\label{sec:treewalking}

\begin{conjecture}
The hierarchy is strict, that is, $\Sigma_n \subsetneq \Sigma_{n+1}$ for all $n$ (and the same is true for all combinations of $\Pi, \Sigma$ and $\Delta$).
\end{conjecture}

We present some elements supporting this conjecture.
Tree-walking automata (TWA) is a similar model that recognise words written on (finite or infinite) trees. For example, Theorem~\ref{thm:soficAlt} can be seen as the translation of \cite{bojanczyk2005}.

In the context of tree-walking automata, we could not find any work on a similar $\Sigma / \Pi$ hierarchy. However, \cite{Cate2010} considers $k$-nested TWA, defined intuitively as follows: $0$-nested TWA are (existential) nondeterministic TWA; $k$-nested TWA are nondeterministic TWA where, at each step, the set of available transitions is determined by foreseeing the behaviour of two $k-1$-nested TWA $A$ and $\overline{A}$, in the following sense: the next transition is chosen nondeterministically among transitions after which $A$ would accept \emph{and} $\overline{A}$ would reject.

The class of (tree-walking or plane-walking) $k$-nested automata seems to be related to $\Sigma_{k+1}$ automata, although we do not have a proof of either direction.
For tree-walking automata, Theorem~29 in \cite{Cate2010} proves that the hierarchy of languages recognized by $k$-nested TWA is strict. This is to us a strong indication that the alternating hierarchy is strict on trees (free groups).

We believe that these results can be brought to plane-walking automata by considering subshifts where trees are drawn on a background on zeroes (this can be ensured with finitely many forbidden patterns), and every tree must belong to $L_k$, where $L_k$ is a tree language that is $\Sigma_{k+1}$ but not $\Sigma_{k}$ (alternatively, $k$-nested and not $k-1$-nested). This is not straightforward as plane-walking automata have the ability to walk out of the tree, which should not provide additional recognition power, but pumping arguments are very tedious for alternating plane-walking automata. We leave this as an open question for future research.

\subsection{An Alternative Approach: Kari-Moore's Rectangles.}\label{sec:squares}

We mention an alternative approach to prove that $\Sigma_1\neq\Pi_1$ that was used in \cite{KariMoore} for finite patterns. The following statements are only translations to our model and from finite patterns to periodic configurations, and we do not vouch for their correctness. This is another suggestion for future work generalising Theorem~\ref{thm:main}.

Let $f : \N\to\mathcal P(\N)$. Let $X_f \subset \{0,1\}^{\Z^2}$ be the smallest subshift containing the strongly periodic configurations generated by the rectangles: 
\[\{ r \in \{0,1\}^{[0,n]\times [0,k]} : k\in f(n) \text{ and for all }i,j>0, r_{0,0} = r_{i,0} = r_{0,j} = 1, r_{i,j} = 0\}.\]

If $X_f$ is accepted by a $\Sigma_1$ plane-walking automaton with $k$ states, then for every $n$, the language $\{1^j : j\in f(n)\}$ is recognised by a two-way nondeterministic finite automaton with $kn$ states, and hence regular \cite{KariMoore}.

The following example of a function $f$ is such that $X_f \in \Sigma_1\backslash\Pi_1$, while $X_{f^{c}} \in \Pi_1\backslash\Sigma_1$:
\[f(n) = \{in+j : i,j\in\N, j<i\}.\]
Indeed, $f^{c}(n)$ is a finite set whose largest element is $(n-2)\cdot n + (n-1) \approx n^2$. If it is accepted by an automaton with $kn$ states, this would be a contradiction for $n>k$ by pumping.

\bibliography{biblio.bib}

\begin{thebibliography}{10}

\bibitem{anselmo10}
Marcella Anselmo and Maria Madonia.
\newblock Classes of two-dimensional languages and recognizability conditions.
\newblock {\em {RAIRO} Theor. Informatics Appl.}, 44(4):471--488, 2010.
\newblock \href {https://doi.org/10.1051/ita/2011003}
  {\path{doi:10.1051/ita/2011003}}.

\bibitem{aubrun2013simulation}
Nathalie Aubrun and Mathieu Sablik.
\newblock Simulation of effective subshifts by two-dimensional subshifts of
  finite type.
\newblock {\em CoRR}, abs/1602.06095:35--63, 2016.
\newblock URL: \url{http://arxiv.org/abs/1602.06095}, \href
  {https://arxiv.org/abs/1602.06095} {\path{arXiv:1602.06095}}, \href
  {https://doi.org/10.48550/arXiv.1602.06095}
  {\path{doi:10.48550/arXiv.1602.06095}}.

\bibitem{blum1967}
Manuel Blum and Carl Hewitt.
\newblock Automata on a 2-dimensional tape.
\newblock In {\em 8th Annual Symposium on Switching and Automata Theory,
  Austin, Texas, USA, October 18-20, 1967}, pages 155--160. IEEE, {IEEE}
  Computer Society, 1967.
\newblock \href {https://doi.org/10.1109/FOCS.1967.6}
  {\path{doi:10.1109/FOCS.1967.6}}.

\bibitem{bojanczyk2005}
Mikolaj Bojanczyk and Thomas Colcombet.
\newblock Tree-walking automata do not recognize all regular languages.
\newblock In Harold~N. Gabow and Ronald Fagin, editors, {\em Proceedings of the
  37th Annual {ACM} Symposium on Theory of Computing, Baltimore, MD, USA, May
  22-24, 2005}, pages 234--243. {ACM}, 2005.
\newblock \href {https://doi.org/10.1145/1060590.1060626}
  {\path{doi:10.1145/1060590.1060626}}.

\bibitem{durand2012fixed}
Bruno Durand, Andrei Romashchenko, and Alexander Shen.
\newblock Fixed-point tile sets and their applications.
\newblock {\em Journal of Computer and System Sciences}, 78(3):731--764, 2012.
\newblock \href {https://doi.org/10.1016/j.jcss.2011.11.001}
  {\path{doi:10.1016/j.jcss.2011.11.001}}.

\bibitem{reco}
Dora Giammarresi and Antonio Restivo.
\newblock Recognizable picture languages.
\newblock {\em International Journal of Pattern Recognition and Artificial
  Intelligence}, 6(2{\&}3):241--256, 1992.
\newblock \href {https://doi.org/10.1142/S021800149200014X}
  {\path{doi:10.1142/S021800149200014X}}.

\bibitem{giamma97}
Dora Giammarresi and Antonio Restivo.
\newblock Two-dimensional languages.
\newblock In Grzegorz Rozenberg and Arto Salomaa, editors, {\em Handbook of
  Formal Languages, Volume 3: Beyond Words}, pages 215--267. Springer, 1997.
\newblock URL: \url{https://doi.org/10.1007/978-3-642-59126-6\_4}, \href
  {https://doi.org/10.1007/978-3-642-59126-6_4}
  {\path{doi:10.1007/978-3-642-59126-6_4}}.

\bibitem{jonoska09}
Nata{\v{s}}a Jonoska and Joni~B. Pirnot.
\newblock Finite state automata representing two-dimensional subshifts.
\newblock {\em Theoretical computer science}, 410(37):3504--3512, 2009.
\newblock \href {https://doi.org/10.1016/j.tcs.2009.03.015}
  {\path{doi:10.1016/j.tcs.2009.03.015}}.

\bibitem{kari2001new}
Jarkko Kari and Cristopher Moore.
\newblock New results on alternating and non-deterministic two-dimensional
  finite-state automata.
\newblock In Afonso Ferreira and Horst Reichel, editors, {\em Annual Symposium
  on Theoretical Aspects of Computer Science}, volume 2010 of {\em Lecture
  Notes in Computer Science}, pages 396--406. Springer, 2001.
\newblock URL: \url{https://doi.org/10.1007/3-540-44693-1\_35}, \href
  {https://doi.org/10.1007/3-540-44693-1_35}
  {\path{doi:10.1007/3-540-44693-1_35}}.

\bibitem{KariMoore}
Jarkko Kari and Cristopher Moore.
\newblock Rectangles and squares recognized by two-dimensional automata.
\newblock In Juhani Karhum{\"{a}}ki, Hermann~A. Maurer, Gheorghe Paun, and
  Grzegorz Rozenberg, editors, {\em Theory is Forever: Essays Dedicated to Arto
  Salomaa on the Occasion of His 70th Birthday}, volume 3113 of {\em Lecture
  Notes in Computer Science}, pages 134--144. Springer, 2004.
\newblock URL: \url{https://doi.org/10.1007/978-3-540-27812-2\_13}, \href
  {https://doi.org/10.1007/978-3-540-27812-2_13}
  {\path{doi:10.1007/978-3-540-27812-2_13}}.

\bibitem{kari2011survey}
Jarkko Kari and Ville Salo.
\newblock A survey on picture-walking automata.
\newblock In Werner Kuich and George Rahonis, editors, {\em Algebraic
  Foundations in Computer Science - Essays Dedicated to Symeon Bozapalidis on
  the Occasion of His Retirement}, volume 7020 of {\em Lecture Notes in
  Computer Science}, pages 183--213. Springer, 2011.
\newblock URL: \url{https://doi.org/10.1007/978-3-642-24897-9\_9}, \href
  {https://doi.org/10.1007/978-3-642-24897-9_9}
  {\path{doi:10.1007/978-3-642-24897-9_9}}.

\bibitem{LM}
Douglas Lind and Brian Marcus.
\newblock {\em An introduction to symbolic dynamics and coding}.
\newblock Cambridge University Press, 1995.

\bibitem{salo2020four}
Ville Salo.
\newblock Four heads are better than three.
\newblock In Hector Zenil, editor, {\em Cellular Automata and Discrete Complex
  Systems - 26th {IFIP} {WG} 1.5 International Workshop, {AUTOMATA} 2020,
  Stockholm, Sweden, August 10-12, 2020, Proceedings}, volume 12286 of {\em
  Lecture Notes in Computer Science}, pages 111--125. Springer, 2020.
\newblock URL: \url{https://doi.org/10.1007/978-3-030-61588-8\_9}, \href
  {https://doi.org/10.1007/978-3-030-61588-8_9}
  {\path{doi:10.1007/978-3-030-61588-8_9}}.

\bibitem{planewalking}
Ville Salo and Ilkka T{\"o}rm{\"a}.
\newblock Plane-walking automata.
\newblock In Teijiro Isokawa, Katsunobu Imai, Nobuyuki Matsui, Ferdinand Peper,
  and Hiroshi Umeo, editors, {\em Cellular Automata and Discrete Complex
  Systems - 20th International Workshop, {AUTOMATA} 2014, Himeji, Japan, July
  7-9, 2014, Revised Selected Papers}, volume 8996 of {\em Lecture Notes in
  Computer Science}, pages 135--148. Springer, 2014.
\newblock URL: \url{https://doi.org/10.1007/978-3-319-18812-6\_11}, \href
  {https://doi.org/10.1007/978-3-319-18812-6_11}
  {\path{doi:10.1007/978-3-319-18812-6_11}}.

\bibitem{groupwalking}
Ville Salo and Ilkka T{\"o}rm{\"a}.
\newblock Group-walking automata.
\newblock In Jarkko Kari, editor, {\em Cellular Automata and Discrete Complex
  Systems: 21st IFIP WG 1.5 International Workshop, AUTOMATA 2015, Turku,
  Finland, June 8-10, 2015. Proceedings 21}, volume 9099 of {\em Lecture Notes
  in Computer Science}, pages 224--237. Springer, 2015.
\newblock URL: \url{https://doi.org/10.1007/978-3-662-47221-7\_17}, \href
  {https://doi.org/10.1007/978-3-662-47221-7_17}
  {\path{doi:10.1007/978-3-662-47221-7_17}}.

\bibitem{Cate2010}
Balder ten Cate and Luc Segoufin.
\newblock Transitive closure logic, nested tree walking automata, and xpath.
\newblock {\em J. {ACM}}, 57(3):18:1--18:41, 2010.
\newblock \href {https://doi.org/10.1145/1706591.1706598}
  {\path{doi:10.1145/1706591.1706598}}.

\bibitem{terrier19}
V{\'e}ronique Terrier.
\newblock Communication complexity tools on recognizable picture languages.
\newblock {\em Theoretical Computer Science}, 795:194--203, 2019.
\newblock \href {https://doi.org/10.1016/j.tcs.2019.05.040}
  {\path{doi:10.1016/j.tcs.2019.05.040}}.

\end{thebibliography}
\end{document}